\documentclass[12pt]{article}
\usepackage{amsmath}
\usepackage{amsthm}
\usepackage{graphicx}
\usepackage{txfonts}
\usepackage[title]{appendix}
\usepackage{enumerate}
\usepackage[colorinlistoftodos]{todonotes}
\usepackage[a4paper, margin=1in]{geometry}
\usepackage{amssymb}
\usepackage{apacite}
\usepackage{listings}
\usepackage{relsize}

\usepackage{xcolor}
\usepackage[colorlinks = true,
            linkcolor = blue,
            urlcolor  = blue,
            citecolor = blue,
            anchorcolor = blue]{hyperref}
\usepackage{float}
\usepackage{bm}
\usepackage{natbib}
\usepackage{times}
\usepackage{caption}
\usepackage{subcaption}
\usepackage[bottom]{footmisc}
\usepackage[plain,noend]{algorithm2e}
\usepackage{float}
\newtheorem{theorem}{Theorem}
\numberwithin{theorem}{section}
\newtheorem{Lemma}{Lemma}
\numberwithin{Lemma}{section}

\numberwithin{Definition}{section}

\allowdisplaybreaks
\usepackage{fancyhdr}
\fancypagestyle{plain}{%
    \fancyhf{}%
    \fancyfoot[R]{\thepage}%
}
\usepackage{color,soul}
\usepackage{array}
\newcolumntype{L}[1]{>{\centering}m{#1}}
\usepackage{float}
\floatstyle{plaintop}
\restylefloat{table}
\usepackage{multirow}
\usepackage{longtable}
\newcommand{\norm}[1]{ \left\| #1 \right\| }

\begin{document}
\title{\large $U$-Statistics Based Jackknife Empirical Likelihood Tests for the Generalized Lorenz Curves}
\author{\small Suthakaran Ratnasingam$^{\footnote{Corresponding author. Email: suthakaran.ratnasingam@csusb.edu}}$, Anton Butenko\\
\small   Department of Mathematics \\ \small California State University, San Bernardino, San Bernardino, CA 92407, USA}
\date{}
\maketitle
\vspace{-0.5cm}
\begin{abstract}
\noindent
A Lorenz curve is a graphical representation of the distribution of income or wealth within a population. The generalized Lorenz curve can be created by scaling the values on the vertical axis of a Lorenz curve by the average output of the distribution. In this paper, we propose two non-parametric methods for testing the equality of two generalized Lorenz curves. Both methods are based on empirical likelihood and utilize a $U$-statistic. We derive the limiting distribution of the likelihood ratio, which is shown to follow a chi-squared distribution with one degree of freedom. We perform simulations to evaluate how well the proposed methods perform compared to an existing method, by examining their Type I error rates and power across different sample sizes and distribution assumptions. Our results show that the proposed methods exhibit superior performance in finite samples, particularly in small sample sizes, and are robust across various scenarios. Finally, we use real-world data to illustrate the methods of testing two generalized Lorenz curves.

\end{abstract}

\noindent
\textbf{Keywords:} generalized Lorenz curve; jackknife empirical likelihood; adjusted empirical likelihood; $U$-statistics; hypothesis testing

\section{Introduction}

A Lorenz curve is a visual representation of an income or wealth distribution within a population. It is named after American economist Max Lorenz (\cite{lor1905}) who developed it in 1905. The Lorenz curve is constructed by plotting the cumulative percentage of the population on the $x-$axis against the cumulative percentage of the variable (such as income or wealth) on the $y-$axis. The resulting curve represents the distribution of the variable in the population. A Lorenz curve that is close to the line of equality (which is a straight line that represents perfect equality) indicates that the distribution of the variable is relatively equal across the population. On the other hand, a Lorenz curve that is farther away from the line of equality indicates a higher degree of inequality in the distribution of the variable. Following \cite{gas1971}, a general definition of the Lorenz curve is given as
\begin{equation}\label{}
\begin{aligned}
\xi(t) = \dfrac{1}{\mu} \int_{0}^{\psi_{t}}x dF(x), \quad t  \in [0,1]
\end{aligned}
\end{equation}
where $\mu$ denotes the mean of $F$, and $\psi_{t} = F^{-1} (t)  = \inf \{x : F(x) \geq t\}$ is the $t-$th quantile of $F$. For a fixed $t \in [0, 1]$, the Lorenz ordinate $\xi(t)$ is the proportion of the cumulative income of the lowest $t$-th quantile of households. The generalized Lorenz curve can be constructed from a Lorenz curve by scaling the values on the vertical axis by the average output of the distribution. Similarly, the generalized Lorenz curve is defined by
\begin{equation}\label{}
\begin{aligned}
\eta(t) = \int_{0}^{\psi_{t}}x dF(x), \quad t  \in [0,1]
\end{aligned}
\end{equation}
where $\psi_{t}$ is the $t-$th quantile of $F$ as defined above. For a fixed $t \in [0, 1]$, the generalized Lorenz ordinate $\eta(t)$ is the average income of the lowest $t$-th quantile of households. While Lorenz curves are frequently used in economics to represent financial inequality, they also can be used in other fields of study to visualize the inequality of the distribution within any system. For example, the Lorenz curve has been used by several researchers to analyze physician distributions. \cite{chang1997} examined variations in the distribution of pediatricians among the states between 1982 and 1992 using Lorenz curves and Gini indices. \cite{kob1992} used the Lorenz curve and the Gini coefficient to study the disparity in physician distribution in Japan.\\

The empirical likelihood (EL) method, introduced by \cite{owen1988}, is a powerful nonparametric approach that offers numerous advantages over traditional methods. Unlike traditional methods, EL does not require strong assumptions to utilize the likelihood ratio approach, yet still preserves many of its desirable features, including Wilk's theorem, asymmetric confidence intervals, and better coverage for small sample sizes. However, the EL method has some computational difficulties when using nonlinear statistics, as demonstrated by \cite{jing2009}, and when constraints' solutions do not exist, as shown by \cite{chen2008}. Specifically, \cite{jing2009} showed that the EL-based approach loses its appeal when using nonlinear $U$-statistics with a degree of $m\geq{2}$ because of the computational difficulty of solving a system of nonlinear equations simultaneously using Lagrange multipliers. To address this issue, \cite{jing2009} proposed the jackknife empirical likelihood (JEL) approach. The JEL method turns the statistic of interest into a sample mean based on jackknife pseudo-values (\cite{que1956}), which are asymptotically independent under mild conditions (\cite{shi1984}). Then, Owen's EL method can be applied consecutively, resulting in a simpler system of equations. On the other hand, \cite{chen2008} pointed out that under certain conditions, it can be challenging to determine the parameter region over which the likelihood ratio function is well-defined. This makes it difficult to identify the maximum likelihood ratio or find a proper initial value. To tackle this challenge, \cite{chen2008} proposed the adjusted empirical likelihood (AEL) approach, which extends the convex hull to include the origin by adding a pseudo-value. With this adjustment, the empirical likelihood is well-defined for all parameter values, making it easier to find the maximum. \\

Multiple studies have been conducted on EL for the Lorenz curve by various researchers. For instance, \cite{belinga2007} and \cite{yang2012} developed plug-in empirical likelihood-based inferences to construct confidence intervals for the generalized Lorenz curve. Most recently, \cite{ratlorenz2023} developed three nonparametric EL-based methods to construct confidence intervals for the generalized Lorenz curve using adjusted empirical likelihood (AEL), transformed empirical likelihood (TEL), and transformed adjusted empirical likelihood (TAEL). Moreover, several studies have focused on comparing two Lorenz curves. For example, \cite{arora2006} investigated the generalized Lorenz dominance and proposed tests for the equality of two generalized Lorenz curves over a specified interval. \cite{li2018} noted that normal approximation-based methods may have poor performance, especially for the skewed income data, or the limiting distributions are nonstandard and bootstrap calibrations are needed hence more effective inferences for Lorenz curves are desirable. \cite{xu1997} proposed an asymptotically distribution-free statistical (ADF) test to evaluate the equality of two generalized Lorenz curves and showed that the test statistic follows the weighted sum of $\chi^{2}$ with different degrees of freedom. \\

As far as we know, there have been no previous studies examining the testing of the equality of two generalized Lorenz curves using EL methods. Therefore, this is the first study to investigate the equality of two generalized Lorenz curves using a nonparametric approach. We propose two novel nonparametric methods that employ a $U$-statistic based on the jackknife empirical likelihood (JEL) method and its extension to the adjusted jackknife empirical likelihood (AJEL). These methods combine two EL-based approaches, namely the JEL and AEL, previously discussed in \cite{jing2009} and \cite{chen2008}, respectively.\\

The remainder of the paper is structured as follows. In Section 2, we present two novel nonparametric techniques for testing the similarity between two generalized Lorenz curves. Section 3 describes the simulation studies carried out to evaluate the effectiveness of the proposed methods in different scenarios and to compare their performance with an existing method. In Section 4, we demonstrate the application of these methods to real datasets. Our findings are discussed in Section 5. All proofs are provided in the appendix.

\section{Main Results}
\sloppy In this section we develop two new testing procedures using jackknife EL methods. Let $X_{1}, X_{2}, \cdots, X_{n_{1}}$ and $Y_{1}, Y_{2}, \cdots, Y_{n_{2}}$ be two random samples from two independent populations. The generalized Lorenz curve for these two samples are
\begin{equation}\label{}
\begin{aligned}
\eta_{1}(t) = \int_{0}^{\psi_{t}}x dF(x), \quad t  \in [0,1]
\end{aligned}
\end{equation}
and
\begin{equation}\label{}
\begin{aligned}
\eta_{2}(t) = \int_{0}^{\psi_{t}}y dF(y), \quad t  \in [0,1]
\end{aligned}
\end{equation}
where $\psi_{t} = F^{-1} (t)  = \inf \{x : F(x) \geq t\}$ is the $t-$th quantile of $F$. We are interested in testing the following hypotheses.
\begin{equation}\label{}
\begin{aligned}
H_{0}: \eta_{1}(t) = \eta_{2}(t) \quad vs \quad H_{1}: \eta_{1}(t) \neq \eta_{2}(t) 
\end{aligned}
\end{equation} 
From the definition of the generalized Lorenz curve, it can be clearly seen that 
\begin{equation*}
\begin{aligned}
E[ X\,I(X \leq \psi_{t})] - \eta_{1}(t) = 0.
\end{aligned}
\end{equation*}
and
\begin{equation*}
\begin{aligned}
E[ Y\,I(Y \leq \psi_{t})] - \eta_{2}(t) = 0.
\end{aligned}
\end{equation*}
As a result, the generalized Lorenz ordinates $\eta_{1}(t)$ and $\eta_{2}(t)$  are the means of the random variable $X$ and $Y$ truncated at $\psi_{t}$ respectively. Let's consider the kernel function,
\begin{equation}\label{}
\begin{aligned}
h(X,Y)
= X\,I(X \leq \psi_{t})   - Y\,I(Y \leq \psi_{t}) 
\end{aligned}
\end{equation} 
We can easily show that $\theta(t) \equiv E\big[h(X_{i},Y_{j} \big] = \eta_{1}(t)  - \eta_{2}(t) $. Thus, we are interested in testing
\begin{equation}\label{hypo}
\begin{aligned}
H_{0}: \theta(t) = 0 \quad vs \quad H_{1}: \theta(t) \neq 0.
\end{aligned}
\end{equation} 
Now consider, the two-sample $U$-statistics of degree (1,1) with the kernel $h$ is given by,
\begin{equation}\label{}
\begin{aligned}
U_{n_{1}, n_{2}} &= \dfrac{1}{n_{1}}\dfrac{1}{n_{2}}\sum_{1\leq i \leq n_{1}}\sum_{1 \leq j \leq n_{2}}h(X_{i},Y_{j}) \\
 &= \dfrac{1}{n_{1}}\dfrac{1}{n_{2}} \sum_{1\leq i \leq n_{1}}\sum_{1 \leq j \leq n_{2}} X_{i}\,I(X_{i} \leq \psi_{t})   - Y_{j}\, I(Y_{j}\leq \psi_{t}) 
\end{aligned}
\end{equation} 
Let $n = n_{1}+n_{2}$. We can write the $U$-statistics 
\begin{equation}\label{}
\begin{aligned}
U_{n_{1}, n_{2}}(X_{1},\dots, X_{n_{1}},Y_{1},\dots,Y_{n_{2}}) = U_{n}(Z_{1},Z_{2},\dots, Z_{n})
\end{aligned}
\end{equation} 
where 
\[ Z_{k} = \begin{cases} 
      X_{k} & k = 1,2,\dots, n_{1} \\
      Y_{k-n_{1}} & k = n_{1}+1,\dots, n
   \end{cases}
\]

We define the corresponding jackknife pseudo-values by 
\begin{equation}\label{eq10}
\begin{aligned}
\widehat{V}_{k} = nU_{n} - (n-1)U_{n-1}^{-k}, \quad k = 1,2,\cdots n,
\end{aligned}
\end{equation} 
where $U_{n-1}^{-k} = U_{n}\big(Z_{1},Z_{2},\dots, Z_{k-1},Z_{k+1},\dots, Z_{n}\big)$. Further, the jackknife estimator of $\theta$ is $n^{-1} \sum_{i=1}^{n} \widehat{V}_{i}$. In particular, under mild conditions, the $\widehat{V}_{k}$'s are asymptotically independent. For more details, readers are referred to \cite{shi1984}. Thus, we can use the EL approach to the $\widehat{V}_{k}$'s. It should be noted that $\widehat{V}_{k}(t)$ is the function of $t$ and can be calculated at a fixed value $t_{0}$ such that $t_{0} \in [0,1]$. For the simplicity of notations, we use $\widehat{V}_{k}$ instead of $\widehat{V}_{k}(t)$.
The JEL for $\theta(t)$ is defined as follows:
\begin{equation}\label{}
\begin{aligned}
L(\theta(t)) = \sup_{\mathbf{p}}\bigg\{\prod_{k=1}^{n}p_{k}: \sum_{k=1}^{n}p_{k}=1, \sum_{k=1}^{n}p_{k} \big(\widehat{V}_{k} - \mathbf{E}\widehat{V}_{k}\big) = 0 \bigg\},
\end{aligned}
\end{equation} 

where $\mathbf{p} = (p_{1}, p_{2}, \dots, p_{n})$ is a probability vector satisfying $\sum_{k=1}^{n}p_{k} = 1$ and $p \geq 0$ for all $k$, and $\mathbf{E}\widehat{V}_{k}$ can be determined using the equation (14) in \cite{jing2017}. Note that $\prod_{k=1}^{n} p_{k}$, subject to $\sum_{k=1}^{n}p_{k} = 1$, attains its maximum $n^{-n}$ at $p_{k} = n^{-1}$. Thus, the JEL ratio for $\theta(t)$ is given as
\begin{equation}\label{}
\begin{aligned}
\mathcal{R} (\theta (t)) = \sup \bigg\{\prod_{k=1}^{n}np_{k}: \sum_{k=1}^{n}p_{k}=1, \sum_{k=1}^{n}p_{k} \big(\widehat{V}_{k} - \mathbf{E}\widehat{V}_{k}\big) = 0 \bigg\}
\end{aligned}
\end{equation}
Further, under null hypothesis $H_{0}: \theta(t) = 0$, the JEL ratio becomes, 
\begin{equation}\label{}
\begin{aligned}
\mathcal{R} (0) = \sup \bigg\{\prod_{k=1}^{n}np_{k}: \sum_{k=1}^{n}p_{k}=1, \sum_{k=1}^{n}p_{k} \widehat{V}_{k} = 0 \bigg\}.
\end{aligned}
\end{equation}

Using the Lagrange multiplier method, we have
\begin{equation*}
\begin{aligned}
p_{k} = \dfrac{1}{n} \bigg\{ 1 + \lambda \widehat{V}_{k}\bigg\}^{-1}, \quad k = 1,\dots, n.
\end{aligned}
\end{equation*}
where $\lambda$ is the solution to
\begin{equation*}
\begin{aligned}
\dfrac{1}{n} \sum_{k=1}^{n} \dfrac{\widehat{V}_{k}}{1 + \lambda \widehat{V}_{k}} = 0.
\end{aligned}
\end{equation*}

Hence, the profile jackknife empirical log-likelihood ratio for $\theta(t)$ becomes
\begin{equation}\label{eqaug3}
\begin{aligned}
\ell(\theta(t)) = -2 \log \mathcal{R} (\theta (t)) = 2 \sum_{k=1}^{n}\log \big\{1 + \lambda \widehat{V}_{k} \big\}.
\end{aligned}
\end{equation}
Let $h_{1,0}(x) = \mathbf{E}h(x,Y_{1}),~~ \sigma_{1,0}^{2} = Var\big(h_{1,0}(X_{1})\big),~~ h_{0,1}(y) = \mathbf{E}h(X_{1},y),$ and $\sigma_{0,1}^{2} = Var\big(h_{0,1}(Y_{1})\big)$. We have the following theorem for the JEL.

\begin{theorem}\label{thm1}
Assume that 
\begin{enumerate}
    \item $E(X^{2})<\infty$, and $E(Y^{2})<\infty$
    \item $\mathbf{E}h^{2}(X_{1},Y_{1})< \infty$, $\sigma_{1,0}^{2}>0$, and $\sigma_{0,1}^{2} > 0$
    \item $n_{1}/n_{2} \longrightarrow r$, where $0<r<\infty$
\end{enumerate}
For any given $t = t_{0} \in (0,1)$, the limiting distribution of $\ell(\theta(t_{0}))$  defined by (\ref{eqaug3}) is a chi-square distribution with one degree of freedom,
\begin{equation}\label{}
\begin{aligned}
\ell(\theta(t_{0})) \longrightarrow \chi^{2}_{1}, \quad \text{as}\,\, \min(n_{1}, n_{2}) \longrightarrow \infty.
\end{aligned}
\end{equation}
\end{theorem}
\begin{proof}
Proof of Theorem \ref{thm1} is given in Appendix.
\end{proof}

Further, \cite{chen2008} proposed the AEL method by adding a pseudo-observation to the data set. This method bypasses the convex hull constraint and ensures a solution at any parameter point. By adopting \cite{chen2008}'s idea, we extend the proposed JEL method by employing the adjusted jackknife empirical likelihood (AJEL) to examine the equality of two generalized Lorenz curves. The AJEL for $\theta(t)$ is defined as follows:
\begin{equation}\label{}
\begin{aligned}
L^{\textnormal{Adj}}(\theta(t)) = \sup_{\mathbf{p}}\bigg\{\prod_{k=1}^{n+1}p^{\textnormal{Adj}}_{k}: \sum_{k=1}^{n+1}p^{\textnormal{Adj}}_{k}=1, \sum_{k=1}^{n+1}p^{\textnormal{Adj}}_{k} g^{\textnormal{Adj}}_{k}(t) = 0 \bigg\},
\end{aligned}
\end{equation} 
where $g^{\textnormal{Adj}}_{k}(t) = \widehat{V}_{k} - \mathbf{E}\widehat{V}_{k},\,\, k = 1,\dots,n$, and $g^{\textnormal{Adj}}_{n+1} =-a_n \Bar{g}_{n}(t) = -\frac{a_n}{n} \sum_{i=1}^{n}g_{i}(t)$. As recommended by \cite{chen2008}, $a_{n}= \max\{1, \log(n)/2\}$. Using the Lagrange multiplier method, we can determine $L^{\textnormal{Adj}}(\theta(t))$ as follows.
\begin{equation*}
\begin{aligned}
p^{\textnormal{Adj}}_{k} = \dfrac{1}{n+1} \bigg\{ 1 + \lambda^{\textnormal{Adj}}(t) g^{\textnormal{Adj}}_{k}(t)\bigg\}^{-1}, \quad k = 1,\dots, n+1.
\end{aligned}
\end{equation*}
where $\lambda^{\textnormal{Adj}}$ is the solution to
\begin{equation*}
\begin{aligned}
\dfrac{1}{n+1} \sum_{k=1}^{n+1} \dfrac{g^{\textnormal{Adj}}_{k}(t)}{1 + \lambda^{\textnormal{Adj}}(t) g^{\textnormal{Adj}}_{k}(t)} = 0.
\end{aligned}
\end{equation*}
Note that $\prod_{k=1}^{n+1} p^{\textnormal{Adj}}_{k}$, subject to $\sum_{k=1}^{n+1}p^{\textnormal{Adj}}_{k} = 1$, attains its maximum $(n+1)^{-n-1}$ at $p_{k} = (n+1)^{-1}$. Thus, the AJEL ratio for $\theta(t)$ is given as
\begin{equation}\label{}
\begin{aligned}
\mathcal{R}^{\textnormal{Adj}} (\theta (t)) = \prod_{k=1}^{n+1}(n+1)p^{\textnormal{Adj}}_{k} = \prod_{k=1}^{n+1} \big\{1 + \lambda^{\textnormal{Adj}}(t) g^{\textnormal{Adj}}_{k}(t) \big\}^{-1}
\end{aligned}
\end{equation}
Hence, the profile-adjusted jackknife empirical log-likelihood ratio for $\theta(t)$ is
\begin{equation}\label{eqaug3022}
\begin{aligned}
\ell^{\textnormal{Adj}}(\theta(t)) = -2 \log \mathcal{R}^{\textnormal{Adj}} (\theta (t)) = 2 \sum_{k=1}^{n+1}\log \big\{1 + \lambda^{\textnormal{Adj}}(t) g^{\textnormal{Adj}}_{k}(t) \big\}.
\end{aligned}
\end{equation}

\begin{theorem}\label{thm2}
Under the same conditions of Theorem \ref{thm1} and for any given $t = t_{0} \in (0,1)$, the limiting distribution of $\ell^{\textnormal{Adj}}(\theta(t_{0}))$ defined by (\ref{eqaug3022}) is a chi-square distribution with one degree of freedom,
\begin{equation}\label{}
\begin{aligned}
\ell^{\textnormal{Adj}}(\theta(t_{0})) \longrightarrow \chi^{2}_{1}, \quad \text{as}\,\, n \longrightarrow \infty.
\end{aligned}
\end{equation}
\end{theorem}
\begin{proof}
Proof of Theorem \ref{thm2} is given in Appendix.
\end{proof}

\section{Simulation Study}
In this section, we conduct a simulation study to evaluate the performance of the proposed testing methods. In our simulation analysis, we take into account Chi-Square, Exponential, and Half-Normal distributions as the overall distribution function $F(x)$ because the majority of income distributions are positively skewed. Under the null hypothesis, we examine the distributions of $\chi^{2}_{4},~Exp(4)$, and $HN(1)$ using different sample sizes $(n_{1}, n_{2})$ such as $(20, 30), (40, 50), (75, 75),$ and $(100, 100)$. We first assess the Type I error probabilities of the ADF, JEL and AJEL methods with a nominal level of $\alpha = 0.05$. Tables \ref{ty1}-\ref{ty3} provide a summary of the findings, including the probabilities of Type I errors (TE) and their corresponding standard errors (SE), whereas Figure \ref{Typeall} depicts the outcomes graphically. The JEL method appears to perform slightly better or similar to the AJEL method. For instance, when using the $\chi^{2}_{4}$ distribution with sample sizes of (20, 30) at $t = 0.1$, the Type I error probability for the ADF method is 0.007 with a standard error of 0.0027, the JEL method is 0.027 with a standard error of 0.0051, and the AJEL method has a probability of 0.075 and a standard error of 0.0083. When using the $\chi^{2}_{4}$ test with sample sizes of 20 and 30 and $t > 0.2$, the AJEL method produces a Type I error rate that is slightly higher than the expected level. The ADF method comes next, with the JEL method following. However, when testing for the $Exp(1)$ distribution, the ADF method results in a Type I error rate that is much lower than the expected level, and the test becomes more conservative for $t>0.2$. When using the $HN(1)$ distribution, the JEL method performs the best among the three methods, while the ADF method performs the worst for all sample sizes. The Type I error probabilities are slightly above the nominal level for small sample sizes, but improve for larger sample sizes and remain within an acceptable range.  \\

% Anton

{\setlength{\tabcolsep}{0.8em}
\begin{table}[H]
\footnotesize
\caption{Type I error (TE) and standard error (SE) comparison of ADF, JEL, and AJEL tests with nominal level $\alpha = 0.05$ when $X, Y\sim \chi_{4}^{2}$}
\centering
\begin{tabular}{cccccccccccccc}
\hline
  & & \multicolumn{2}{c}{ADF} & \multicolumn{2}{c}{JEL} & \multicolumn{2}{c}{AJEL}    \\  \hline
\multirow{1}{*}{$(n_{1}, n_{2})$}& \multirow{1}{*}{$t$} 	& TE &  SE & TE &  SE & TE &  SE  \\  \hline
(20, 30)&	0.1&	0.007& 0.0027 &			0.027&	0.0051&	0.075&	0.0083 \\
&	0.2&	0.011&	0.0033&	0.019&	0.0043&	0.058&	0.0074 \\
&	0.3&	0.027&	0.0051&	0.018&	0.0042&	0.062&	0.0076 \\
&	0.4&	0.037&	0.0060&	0.017&	0.0041&	0.065&	0.0078 \\
&	0.5&	0.069&	0.0080&	0.017&	0.0041&	0.065&	0.0078 \\
&	0.6&	0.062&	0.0076&	0.059&	0.0075&	0.086&	0.0089 \\
&	0.7&	0.069&	0.0080&	0.052&	0.0070&	0.070&	0.0081 \\
&	0.8&	0.071&	0.0081&	0.046&	0.0066&	0.065&	0.0078 \\
&	0.9&	0.070&	0.0081&	0.038&	0.0060&	0.052&	0.0075 \\ \hline
(40,50)	&0.1	& 0.005	& 0.0023 &		0.016&	0.0040&	0.055&	0.0072 \\
	&0.2&	0.009	&0.0030&	0.019&	0.0042&	0.039&	0.0061 \\
&	0.3&	0.014&	0.0037&	0.010&	0.0031&	0.044&	0.0065 \\
&	0.4&	0.027&	0.0051&	0.013&	0.0036&	0.054&	0.0071 \\
&	0.5&	0.047&	0.0067&	0.022&	0.0046&	0.061&	0.0076 \\
&	0.6&	0.044&	0.0065&	0.049&	0.0068&	0.072&	0.0082 \\
&	0.7&	0.051&	0.0070&	0.050&	0.0069&	0.065&	0.0078 \\
&	0.8	&0.059&	0.0075&	0.042&	0.0063	&0.065&	0.0078 \\
&	0.9&	0.060&	0.0075&	0.043&	0.0063&	0.061&	0.0076 \\ \hline 
(75,75)	&0.1	& 0.001	& 0.0012 &		0.010&	0.0029	&0.012&	0.0034 \\
&	0.2&	0.002&	0.0014	&0.011	&0.0033&	0.013&	0.0039 \\
&	0.3&	0.015&	0.0038&	0.019&	0.0042&	0.016&	0.0040 \\
	&0.4&	0.016&	0.0040&	0.012&	0.0034&	0.018&	0.0048 \\
&	0.5&	0.022&	0.0046&	0.023&	0.0047&	0.024&	0.0051 \\
&	0.6&	0.024&	0.0048&	0.037&	0.0060&	0.046&	0.0073 \\
&	0.7&	0.025&	0.0049&	0.031&	0.0056&	0.031&	0.0056 \\
&	0.8&	0.028&	0.0052&	0.030&	0.0054	&0.035&	0.0071 \\
&	0.9	&0.036&	0.0059&	0.032&	0.0056&	0.031&	0.0056 \\ \hline 
(100, 100)&	0.1		& 0.002& 0.0016 &	0.011&	0.0030&	0.038&	0.0060 \\
&	0.2&	0.006	&0.0024&	0.012&	0.0034&	0.040&	0.0062 \\
&	0.3&	0.014	&0.0037&	0.020&	0.0044&	0.048&	0.0068 \\
&	0.4&	0.024&	0.0048&	0.015&	0.0038&	0.046&	0.0066 \\
&	0.5&	0.033&	0.0056&	0.010&	0.0031&	0.044&	0.0065 \\
&	0.6&	0.036	&0.0059&	0.031&	0.0055&	0.046&	0.0066 \\
&	0.7	&0.038&	0.0060&	0.035&	0.0058&	0.046&	0.0066 \\
&	0.8	&0.047&	0.0067&	0.035&	0.0058&	0.046&	0.0066 \\
&	0.9	&0.045&	0.0066&	0.033&	0.0057&	0.045&	0.0065 \\ \hline 

\end{tabular}
\label{ty1}
\end{table}}

\newpage

{\setlength{\tabcolsep}{0.8em}
\begin{table}[H]
\footnotesize
\caption{Type I error and standard error comparison of ADF, JEL, and AJEL tests with nominal level $\alpha = 0.05$ when $X, Y\sim Exp(4)$}
\centering
\begin{tabular}{cccccccccccccc}
\hline
  & &  \multicolumn{2}{c}{ADF} &  \multicolumn{2}{c}{JEL} & \multicolumn{2}{c}{AJEL}    \\  \hline
\multirow{1}{*}{$(n_{1}, n_{2})$}& \multirow{1}{*}{$t$} 	&  TE &  SE & TE &  SE & TE &  SE  \\ \hline
(20,30)	&0.1	& 0.037	& 0.0060 &	0.081&	0.0086	&0.091&	0.0091 \\
&	0.2&	0.022&	0.0046&	0.047&	0.0067&	0.060	&0.0075 \\
&	0.3&	0.010&	0.0031&	0.074&	0.0083&	0.075&	0.0083 \\
&	0.4&	0.006&	0.0024&	0.061&	0.0076&	0.074&	0.0083 \\
&	0.5&	0.006&	0.0024&	0.060&	0.0075&	0.061&	0.0076 \\
&	0.6&	0.001&	0.0010&	0.096&	0.0093&	0.101&	0.0095 \\
&	0.7&	0.001&	0.0010&	0.088&	0.0090&	0.087&	0.0089 \\
&	0.8&	0.001&	0.0010&	0.065&	0.0078&	0.075&	0.0083 \\
&	0.9&	0.000&	0.0000	&0.062&	0.0076&	0.064&	0.0078 \\ \hline
 
(40,50)&	0.1	& 0.024	& 0.0049 &	0.060&	0.0075&	0.076&	0.0084 \\
&	0.2&	0.010&	0.0031	&0.043	&0.0064&	0.052&	0.0070 \\
&	0.3&	0.000&	0.0000	&0.048	&0.0068&	0.049&	0.0068 \\
&	0.4&	0.000&	0.0000	&0.045	&0.0066&	0.052&	0.0070 \\
&	0.5&	0.000&	0.0000	&0.057&	0.0073&	0.056&	0.0073 \\
&	0.6&	0.000&	0.0000	&0.068	&0.0080&	0.070&	0.0081 \\
&	0.7&	0.000&	0.0000	&0.065&	0.0078&	0.064&	0.0077 \\
&	0.8&	0.000&	0.0000	&0.067&	0.0079&	0.072&	0.0082 \\
&	0.9&	0.000&	0.0000	&0.059&	0.0074&	0.061&	0.0075 \\ \hline

(75,75)&	0.1	& 0.018	& 0.0043 &	0.062&	0.0076&	0.071&	0.0081 \\
&	0.2&	0.005&	0.0022&	0.050	&0.0069&	0.055&	0.0072 \\
&	0.3&	0.000&	0.0000	&0.049&	0.0058&	0.050&	0.0059 \\
&	0.4&	0.000&	0.0000&	0.051&	0.0070&	0.056&	0.0073 \\
&	0.5&	0.000&	0.0000&	0.056&	0.0074&	0.060&	0.0076 \\
&	0.6&	0.000&	0.0000&	0.069&	0.0080&	0.069&	0.0080 \\
&	0.7&	0.000&	0.0000&	0.052&	0.0072&	0.052&	0.0072 \\
&0.8& 0.000&	0.0000& 0.054 &0.0071& 0.056& 0.0073 \\
&	0.9&	0.000&	0.0000&	0.052&	0.0072&	0.055&	0.0074 \\ \hline

(100, 100)&	0.1	& 0.012	& 0.0035 &	0.062&	0.0076&	0.071	&0.0081 \\
	&0.2&	0.006&	0.0024&	0.049&	0.0068&	0.055&	0.0072 \\
	&0.3&	0.001&	0.0010&	0.063&	0.0077&	0.061	&0.0076 \\
	&0.4	&0.001&	0.0010&	0.054&	0.0071&	0.059&	0.0075 \\
	&0.5&	0.000&	0.0000&	0.054&	0.0071&	0.051&	0.0070 \\
	&0.6&	0.000&	0.0000&	0.051&	0.0070&	0.051&	0.0070 \\
	&0.7&	0.000&	0.0000&	0.057&	0.0073&	0.057&	0.0073 \\
& 0.8& 	0.000&	0.0000& 0.057& 0.0073& 0.057& 0.0073 \\
	&0.9&	0.000&	0.0000&	0.052&	0.0072&	0.052&	0.0072 \\ \hline
\end{tabular}
\label{ty2}
\end{table}}

\newpage

{\setlength{\tabcolsep}{0.8em}
\begin{table}[H]
\footnotesize
\caption{Type I error and standard error comparison of ADF, JEL, and AJEL tests with nominal level $\alpha = 0.05$ when $X, Y\sim HN(1)$}
\centering
\begin{tabular}{cccccccccccccc}
\hline
  & &  \multicolumn{2}{c}{ADF} & \multicolumn{2}{c}{JEL} & \multicolumn{2}{c}{AJEL}    \\  \hline
\multirow{1}{*}{$(n_{1}, n_{2})$}& \multirow{1}{*}{$t$} 	& TE &  SE & TE &  SE & TE &  SE  \\ \hline
(20,30)	&0.1	&	0.068	&  0.0080 & 0.074	&0.0083&	0.094&	0.0092 \\
&	0.2&	0.076&	0.0084&	0.063&	0.0077&	0.065&	0.0078 \\
	&0.3&	0.084	&0.0088&	0.050&	0.0069	&0.047&	0.0067 \\
&	0.4&	0.087&	0.0089&	0.051&	0.0070&	0.052&	0.0070 \\
&	0.5&	0.129&	0.0106&	0.050&	0.0069&	0.056&	0.0073 \\
&	0.6 &	0.115	&0.0101&	0.077&	0.0084&	0.090&	0.0090 \\
&	0.7	&0.099&	0.0094&	0.063&	0.0077&	0.069&	0.0080 \\
&	0.8&	0.108&	0.0098&	0.062&	0.0076&	0.078&	0.0085 \\
&	0.9&	0.103&	0.0096&	0.041&	0.0063&	0.071&	0.0081 \\ \hline

(40,50)	&0.1	& 0.062	& 0.0076 &	0.056	&0.0073	&0.066&	0.0079 \\
	&0.2&	0.066&	0.0079	&0.037&	0.0060&	0.042&	0.0063 \\
&	0.3&	0.068&	0.0080	&0.049&	0.0068&	0.054&	0.0071 \\
&	0.4&	0.081&	0.0086	&0.065&	0.0078&	0.058&	0.0074 \\
&	0.5&	0.101&	0.0095	&0.049&	0.0068&	0.059&	0.0075 \\
&	0.6&	0.094&	0.0092&	0.080&	0.0086	&0.087&	0.0089 \\
&	0.7&	0.096&	0.0093	&0.066&	0.0079	&0.064&	0.0077 \\
&	0.8&	0.108&	0.0098&	0.053	&0.0071&	0.070&	0.0081 \\
&	0.9&	0.098&	0.0094&	0.034&	0.0057&	0.055&	0.0072 \\ \hline 

(75,75)&	0.1	& 0.038	& 0.0061 &	0.051&	0.0070&	0.053&	0.0071 \\ 
&	0.2&	0.045&	0.0066&	0.055&	0.0072&	0.054	&0.0071 \\
&	0.3&	0.059&	0.0075&	0.051	&0.0058	& 0.058	&0.0063 \\
&	0.4&	0.058&	0.0074&	0.053&	0.0071&	0.053	&0.0071 \\
&	0.5&	0.077&	0.0084&	0.044&	0.0048&	0.048&	0.0052 \\
&	0.6&	0.079&	0.0085&	0.057&	0.0073&	0.065&	0.0078 \\
&	0.7&	0.081&	0.0086&	0.045&	0.0050&	0.048&	0.0052 \\
&	0.8&	0.091&	0.0091&	0.055&	0.0072	&0.061&	0.0076 \\
&	0.9&	0.098&	0.0094&	0.036&	0.0059	&0.062&	0.0076 \\ \hline

(100, 100)&	0.1	& 0.049	& 0.0069 &	0.046&	0.0066&	0.047&	0.0067 \\
&	0.2&	0.055&	0.0072&	0.071&	0.0081&	0.064&	0.0077 \\
	&0.3&	0.060&	0.0075&	0.054&	0.0071&	0.056&	0.0073 \\
&	0.4&	0.067&	0.0079&	0.060&	0.0075&	0.059&	0.0075 \\
&	0.5&	0.072&	0.0082&	0.048&	0.0068&	0.053&	0.0071 \\
&	0.6&	0.078&	0.0085&	0.065&	0.0078&	0.071&	0.0081 \\
&	0.7&	0.082&	0.0087&	0.066&	0.0079&	0.066&	0.0079 \\
&	0.8&	0.087&	0.0089&	0.063&	0.0077	&0.072&	0.0082 \\
&	0.9	&0.086&	0.0089&	0.039&	0.0061&	0.057&	0.0073 \\ \hline

\end{tabular}
\label{ty3}
\end{table}}

\begin{figure}[H]
  \centering
  \includegraphics[width=1\textwidth]{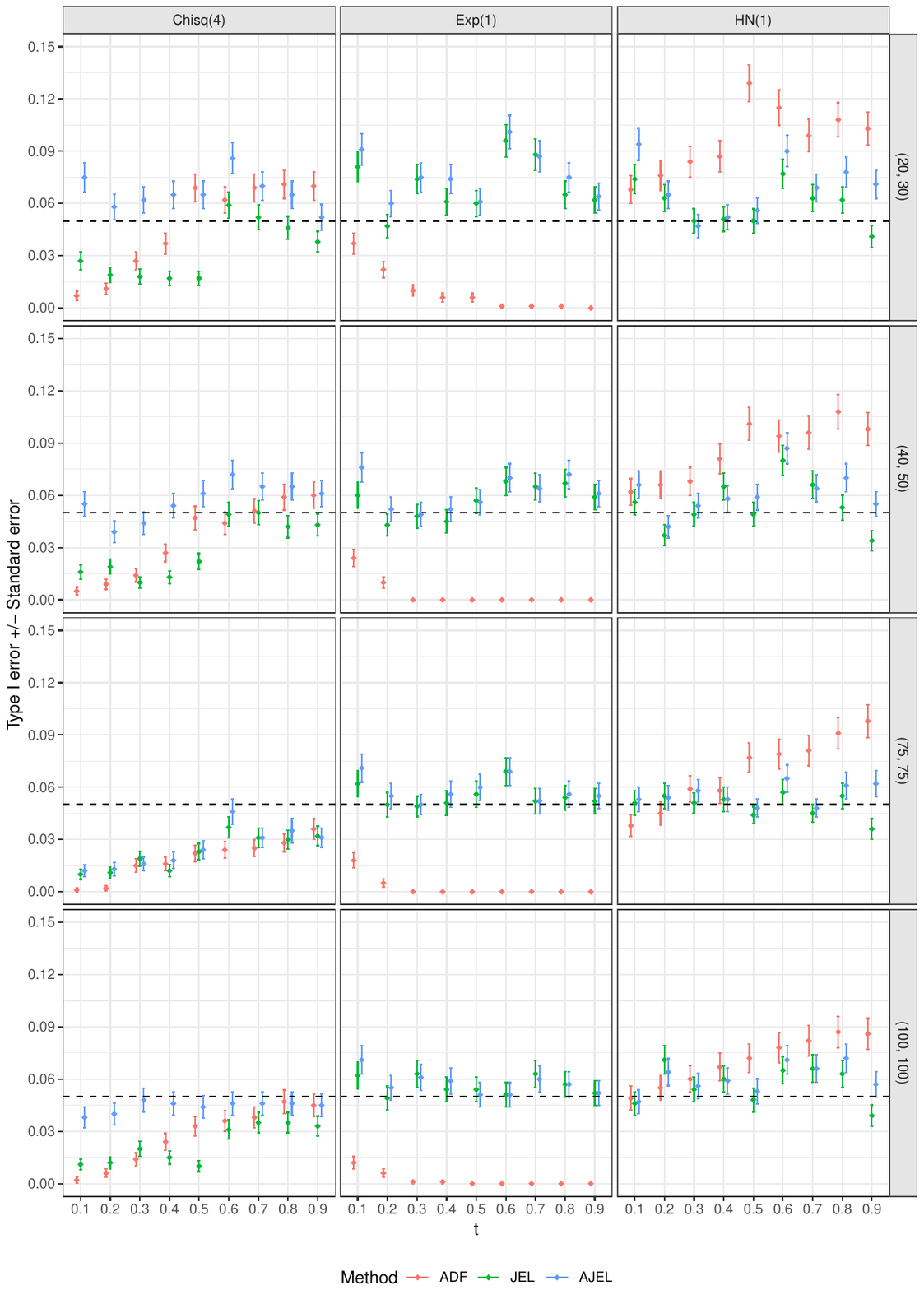}
  \caption{Type I error comparison for ADF, JEL, and AJEL methods for different distributions, sample sizes, and values of $t$}
  \label{Typeall}
\end{figure}

\newpage
Next, we conduct a power analysis for the ADF, JEL, and AJEL methods. Figure \ref{figall} displays the generalized Lorenz curves for Chi-Square, Exponential, and Half-Normal distribution under two sets of parameters. The difference between the generalized Lorenz curves of $\chi^{2}(4)$ and $\chi^{2}(5.5)$ increases as $t$ changes from 0 to 0.5, then the difference decreases as $t$ changes from 0.5 to 1. The difference between the generalized Lorenz curves for $Exp(2)$ and $Exp(4)$ increases significantly as $t$ changes from 0 to 1, while the difference between the generalized Lorenz curves for $HN(1)$ and $HN(1.5)$ also increases but not as significant as in the case of the Exponential distributions. Tables \ref{pw1}-\ref{pw3} present a summary of the outcomes, depicting the power and standard errors, whereas Figure \ref{powerall} illustrates the results graphically. As expected, the power for Chi-Square distributions tends to increase as the value of $t$ ranges from 0.1 to 0.5, followed by a slight drop as $t$ ranges from 0.5 to 0.9. The AJEL method exhibits better power among the three methods, while the ADF method is the weakest. Moreover, for Exponential distributions, the power of the JEL and AJEL methods tends to increase as $t$ ranges from 0 to 0.9, while the power of the ADF method tends to decrease as the value of $t$ changes from 0 to 0.9. When considering the Half-Normal distributions, all three methods show a similar pattern, with the ADF method being the most effective when $t \leq 0.5$, and the JEL and AJEL methods being superior when $t > 0.5$. The increase in power is more pronounced for Exponential distributions than for Half-Normal distributions in the case of the JEL and AJEL methods. In all three cases, the AJEL method outperforms the JEL method slightly.
% Anton

\begin{figure}[H]
  \centering
  \includegraphics[width=1\textwidth]{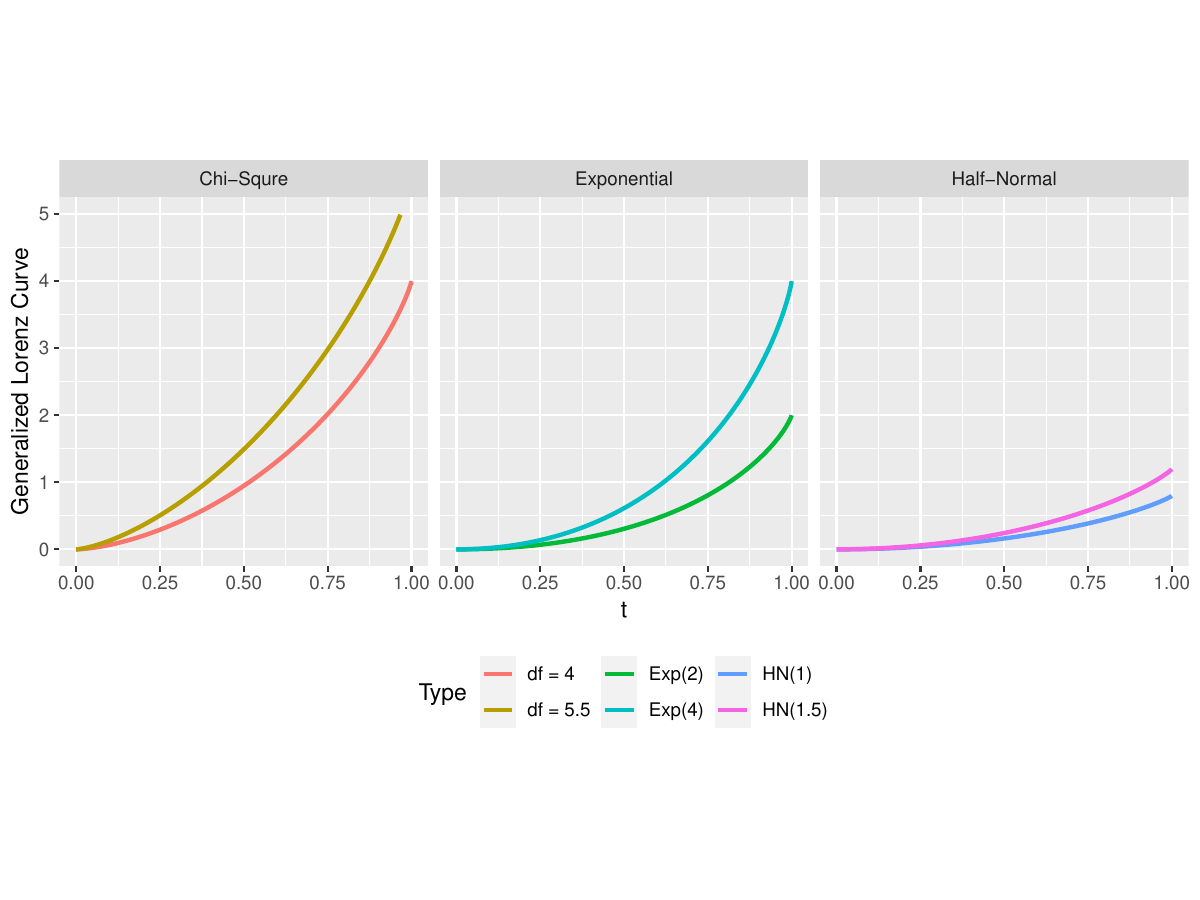}
  \caption{Generalized Lorenz Curves for Chi-Square, Exponential, and Half-Normal Distributions}
  \label{figall}
\end{figure}

{\setlength{\tabcolsep}{0.8em}
\begin{table}[H]
\footnotesize
\caption{Power comparison of ADF, JEL, and AJEL tests with nominal level $\alpha = 0.05$ when $X\sim \chi^{2}_{4}$ and $Y\sim \chi^{2}_{5.5}$}
\centering
\begin{tabular}{ccccccccccccccccc}
\hline
  & &  \multicolumn{2}{c}{ADF} & \multicolumn{2}{c}{JEL} & \multicolumn{2}{c}{AJEL}    \\  \hline
\multirow{1}{*}{$(n_{1}, n_{2})$}& \multirow{1}{*}{$t$} 	&   Power &  SE & Power &  SE & Power &  SE  \\
\hline
(20,30)	&0.1	& 0.128 & 0.0106	&	0.357&	0.0152	&0.372&	0.0153 \\
&	0.2&  0.190&	0.0124&	0.400	&0.0155&	0.415&	0.0156 \\
&	0.3&	0.330&	0.0149&	0.423&	0.0156&	0.441&	0.0157 \\
&	0.4&	0.360&	0.0152&	0.454&	0.0157&	0.471&	0.0158 \\
&	0.5&	0.454&	0.0157&	0.443&	0.0157&	0.461&	0.0158 \\
&	0.6&	0.475&	0.0158&	0.834&	0.0118&	0.840&	0.0116 \\
&	0.7&	0.439&	0.0157&	0.771&	0.0133&	0.781	&0.0131 \\
&	0.8&	0.450&	0.0157&	0.717&	0.0142&	0.725&	0.0141 \\
&	0.9&	0.488&	0.0158&	0.594&	0.0155&	0.610	&0.0154 \\ \hline
 
(40,50)	&0.1 & 0.182	& 0.0122	&	0.483&	0.0158&	0.497&	0.0158 \\
	&0.2&	0.369&	0.0153&	0.595&	0.0155&	0.607&	0.0154 \\
&	0.3&	0.508&	0.0158&	0.660&	0.0150&	0.673&	0.0148 \\
&	0.4&	0.578&	0.0156&	0.664&	0.0149&	0.676&	0.0148 \\
&	0.5&	0.657&	0.0150&	0.644&	0.0151&	0.654&	0.0150 \\
&	0.6&	0.646&	0.0151&	0.935&	0.0078&	0.936&	0.0077 \\
&	0.7&	0.655&	0.0150&	0.914&	0.0089&	0.915&	0.0088 \\
&	0.8&	0.670	&0.0149&	0.878&	0.0103&	0.883&	0.0102 \\
	&0.9&	0.694&	0.0146&	0.783&	0.0130&	0.793&	0.0128 \\ \hline
 
(75,75)&	0.1	 & 0.438 & 0.0157	&	0.688	&0.0147&	0.695&	0.0146 \\
&	0.2&	0.595&	0.0155&	0.811&	0.0124&	0.820&	0.0036 \\
&0.3&	0.757&	0.0136&	0.986&	0.0037	&0.987	&0.0036 \\
&	0.4&	0.790&	0.0129&	0.869&	0.0107&	0.872&	0.0106 \\
&	0.5&	0.841&	0.0116&	0.983&	0.0041&	0.983&	0.0041 \\
&	0.6&	0.852&	0.0112&	0.979&	0.0045&	0.980&	0.0044 \\
&	0.7&	0.856&	0.0111&	0.791&	0.0129&	0.797&	0.0127 \\
&	0.8&	0.872&	0.0106&	0.954&	0.0066&	0.957&	0.0064 \\
&	0.9&	0.875&	0.0105&	0.927&	0.0082&	0.928&	0.0082 \\ \hline

(100,100)&	0.1	& 0.625 & 0.0153	&	0.809&	0.0124&	0.813&	0.0123 \\
	&0.2&	0.786&	0.0130&	0.899&	0.0095&	0.901&	0.0094 \\
&	0.3&	0.871&	0.0106&	0.931&	0.0080&	0.933&	0.0079 \\
	&0.4&	0.923&	0.0084&	0.932&	0.0080&	0.933&	0.0079 \\
&	0.5&	0.932	&0.0080&	0.936&	0.0077&	0.94&	0.0075 \\
	&0.6&	0.933	&0.0079	&0.993&	0.0026&	0.993&	0.0026 \\
	&0.7&	0.942&	0.0074	&0.991&	0.0030&	0.991&	0.0030 \\
	&0.8&	0.951&	0.0068&	0.984&	0.0040&	0.986&	0.0037 \\
	&0.9	&0.953&	0.0067&	0.972	&0.0052&	0.973&	0.0051 \\ \hline
\end{tabular}
\label{pw1}
\end{table}}

\newpage

{\setlength{\tabcolsep}{0.8em}
\begin{table}[H]
\footnotesize
\caption{Power comparison of ADF, JEL, and AJEL tests with nominal level $\alpha = 0.05$ when $X\sim Exp(4)$ and $Y\sim Exp(2)$}
\centering
\begin{tabular}{cccccccccccccc}
\hline
  & &   \multicolumn{2}{c}{ADF} & \multicolumn{2}{c}{JEL} & \multicolumn{2}{c}{AJEL}    \\  \hline
\multirow{1}{*}{$(n_{1}, n_{2})$}& \multirow{1}{*}{$t$} 	& Power &  SE &  Power &  SE & Power &  SE  \\
\hline
(20, 30)&	0.1&	0.695 & 0.0146 &		0.363&	0.0152&	0.382&	0.0154 \\
&	0.2&	0.678&	0.0148&	0.463&	0.0158&	0.478&	0.0158 \\
&	0.3&	0.520&	0.0158&	0.551&	0.0157&	0.566&	0.0157 \\
&	0.4&	0.417&	0.0156&	0.613&	0.0154&	0.626&	0.0153 \\
&	0.5&	0.418&	0.0156&	0.664&	0.0149&	0.682&	0.0147 \\
&	0.6&	0.367&	0.0152&	0.893&	0.0098&	0.901&	0.0094 \\
&	0.7&	0.295&	0.0144&	0.884&	0.0101&	0.894&	0.0097 \\
&	0.8&	0.284&	0.0143&	0.868&	0.0107&	0.875&	0.0105 \\
&	0.9&	0.238&	0.0135&	0.834&	0.0118&	0.846&	0.0114 \\ \hline
							
(40, 50)&	0.1	&	0.835 & 0.0117 &	0.471&	0.0158&	0.485&	0.0158 \\
	&0.2&	0.812&	0.0124&	0.606&	0.0155&	0.615&	0.0154 \\
&	0.3&	0.582&	0.0156&	0.724&	0.0141&	0.730&	0.0140 \\
	&0.4&	0.492&	0.0158&	0.792&	0.0128&	0.803&	0.0126 \\
&	0.5&	0.470&	0.0158&	0.837&	0.0117&	0.851&	0.0113 \\
&	0.6	&0.368&	0.0153&	0.969&	0.0055&	0.971&	0.0053 \\
&	0.7	&0.324&	0.0148&	0.971&	0.0053&	0.972&	0.0052 \\
&	0.8&	0.317&	0.0147&	0.961&	0.0061&	0.966&	0.0057 \\
&	0.9	&0.290&	0.0143&	0.958&	0.0063&	0.960&	0.0062 \\ \hline
							
(75, 75)&	0.1&	0.901 & 0.0094 &		0.616&	0.0154&	0.623&	0.0153 \\
	&0.2&	0.886&	0.0101&	0.787&	0.0129&	0.791&	0.0129 \\
&	0.3&	0.719&	0.0142	&0.975&	0.0049&	0.976&	0.0048 \\
&	0.4&	0.586&	0.0156&	0.930&	0.0081&	0.930&	0.0081 \\
&	0.5&	0.546&	0.0157&	0.993&	0.0026&	0.993&	0.0026 \\
&	0.6&	0.440&	0.0157&	0.995&	0.0022&	0.995&	0.0022 \\
&	0.7&	0.408&	0.0155&	0.987&	0.0036&	0.988&	0.0034 \\
&	0.8&	0.371&	0.0153&	0.995&	0.0022&	0.995&	0.0022 \\
&	0.9&	0.348&	0.0151&	0.994&	0.0024&	0.995&	0.0022 \\ \hline
							
(100, 100)	&0.1 & 0.946 & 0.0071	&		0.661&	0.0150&	0.667&	0.0149 \\
&	0.2&	0.930&	0.0081&	0.851&	0.0113&	0.856&	0.0111 \\
	&0.3&	0.717&	0.0142&	0.927&	0.0082&	0.928&	0.0082 \\
&	0.4	&0.629&	0.0153&	0.965&	0.0058&	0.968&	0.0056 \\
&	0.5&	0.551&	0.0157&	0.984&	0.0040&	0.985&	0.0038 \\
&	0.6&	0.483&	0.0158&	1.000&	0.0000&	1.000&	0.0000 \\
	&0.7	&0.429&	0.0157&	1.000&	0.0000&	1.000&	0.0010 \\
& 0.8&	0.419&	0.0156	&0.999&	0.0010&0.999&	0.0010\\
&	0.9&	0.355	&0.0151&	0.999&	0.0010&	0.999&	0.0010 \\ \hline

\end{tabular}
\label{pw2}
\end{table}}

\newpage

{\setlength{\tabcolsep}{0.8em}
\begin{table}[H]
\footnotesize
\caption{Power comparison of ADF, JEL, and AJEL tests with nominal level $\alpha = 0.05$ when $X\sim HN(1)$ and $Y\sim HN(1.5)$}
\centering
\begin{tabular}{cccccccccccccc}
\hline
  & & \multicolumn{2}{c}{ADF} &  \multicolumn{2}{c}{JEL} & \multicolumn{2}{c}{AJEL}    \\  \hline
\multirow{1}{*}{$(n_{1}, n_{2})$}& \multirow{1}{*}{$t$} 	& Power &  SE & Power &  SE & Power &  SE  \\
\hline
(20, 30)	&0.1& 0.415	&  0.0156 &		0.260&	0.0139&	0.274&	0.0141 \\
&	0.2&	0.438&	0.0157&	0.292&	0.0144&	0.303&	0.0145 \\
&	0.3&	0.479&	0.0158&	0.309&	0.0146&	0.327&	0.0148 \\
&	0.4&	0.485&	0.0158&	0.373&	0.0153&	0.390&	0.0154 \\
&	0.5&	0.541&	0.0158&	0.377&	0.0153&	0.395&	0.0155 \\
&	0.6&	0.559&	0.0157&	0.721&	0.0142&	0.737&	0.0139 \\
&	0.7&	0.520&	0.0158&	0.694&	0.0146&	0.716&	0.0143 \\
&	0.8&	0.537&	0.0158&	0.670&	0.0149&	0.688&	0.0147 \\
&	0.9&	0.548&	0.0157&	0.597&	0.0155&	0.616&	0.0154 \\\hline
							
(40, 50)&	0.1	& 0.654	& 0.0150 &	0.312&	0.0147&	0.330&	0.0149 \\
&	0.2&	0.661	&0.0150&	0.393&	0.0154&	0.400&	0.0155 \\
	&0.3&	0.676&	0.0148&	0.455&	0.0157&	0.468&	0.0158 \\
&	0.4&	0.690&	0.0146&	0.544&	0.0158&	0.561&	0.0157 \\
&	0.5&	0.725&	0.0141&	0.577&	0.0156&	0.589&	0.0156 \\
&	0.6&	0.713&	0.0143&	0.838&	0.0117&	0.845&	0.0114 \\
&	0.7&	0.723&	0.0142&	0.840&	0.0116	&0.847&	0.0114 \\
&	0.8&	0.733&	0.0140&	0.841&	0.0116&	0.847&	0.0114 \\
&	0.9&	0.746&	0.0138&	0.809&	0.0124&	0.816&	0.0123 \\\hline
							
(75, 75)&	0.1& 0.839	& 0.0116 &		0.381&	0.0154&	0.390&	0.0154 \\
&	0.2	&0.843&	0.0115&	0.508&	0.0158&	0.515&	0.0158 \\
	&0.3&	0.853	&0.0112&	0.896&	0.0097&	0.897&	0.0096 \\
&	0.4&	0.866	&0.0108&	0.689&	0.0146&	0.700&	0.0145 \\
&	0.5&	0.877&	0.0104&	0.912&	0.0090&	0.913&	0.0089 \\
&	0.6&	0.876&	0.0104&	0.933&	0.0079&	0.936&	0.0077 \\
&	0.7&	0.882	&0.0102	&0.846&	0.0114	&0.851&	0.0113 \\
	&0.8&	0.891&	0.0099&	0.932&	0.0080&	0.934&	0.0079 \\
&	0.9&	0.892&	0.0098&	0.929&	0.0081&	0.931&	0.0080 \\\hline
							
(100, 100)	&0.1 & 0.938 & 0.0076	&0.390&	0.0155&	0.400&	0.0155 \\
&	0.2&	0.940&	0.0075&	0.554&	0.0157	&0.561&	0.0157 \\
&	0.3&	0.941&	0.0075&	0.691&	0.0146&	0.695&	0.0146 \\
&	0.4&	0.944&	0.0073&	0.787&	0.0129&	0.792&	0.0128 \\
&	0.5&	0.951	&0.0068&	0.837&	0.0117&	0.839&	0.0116 \\
&	0.6&	0.953&	0.0067&	0.958&	0.0063&	0.958&	0.0063 \\
&	0.7&	0.952&	0.0068&	0.969&	0.0055&	0.970&	0.0054 \\
&	0.8&	0.957&	0.0064&	0.973&	0.0051&	0.973&	0.0051 \\
&	0.9&	0.958&	0.0063&	0.972&	0.0052&	0.975&	0.0049 \\ \hline

\end{tabular}
\label{pw3}
\end{table}}
\newpage

\begin{figure}[H]
  \centering
  \includegraphics[width=1\textwidth]{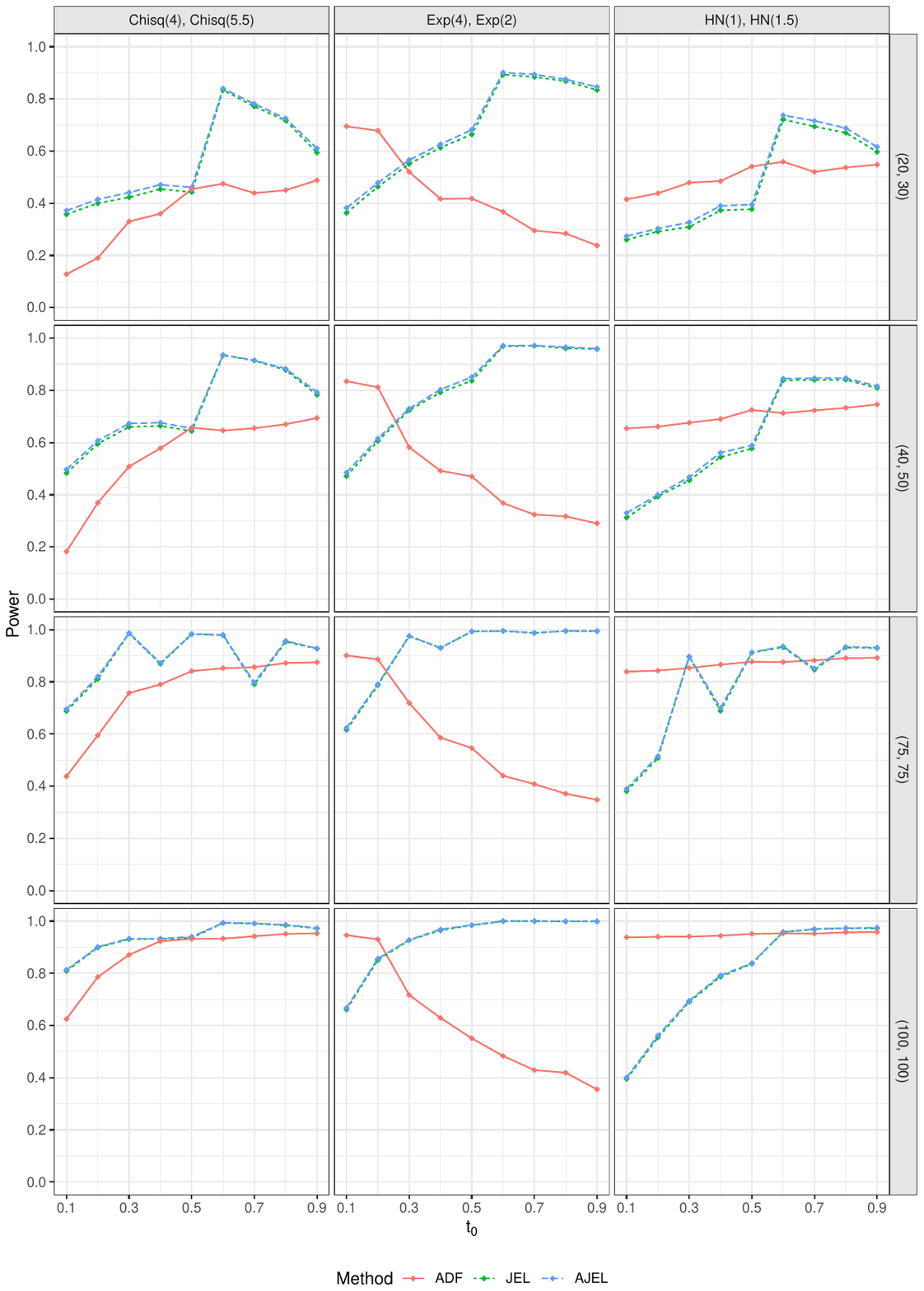}
  \caption{Power comparison for ADF, JEL, and AJEL methods for different distributions, sample sizes, and values of $t$}
  \label{powerall}
\end{figure}

\newpage

\section{Applications}
In this section, we use the proposed methods to evaluate the equality of the generalized Lorenz curves for various subgroups of employees of California State University (CSU) and University of California (UC) systems in 2021.\footnote{The most recent data was obtained and is available on the State Controller's Office website at \url{https://publicpay.ca.gov/Reports/Explore.aspx}.} The 2021 data comprises 105,414 records of salaries for CSU and 299,448 records of salaries for UC. The data is anonymous but grouped by employer name and type of position.\footnote{To identify CSU instructional faculty, we filtered the data set based on the following keywords: ``Instructional Faculty", ``Teaching Associate", ``Visiting Faculty", ``Lecturer", "Academic-Related", ``Department Chair". To identify UC instructional faculty, we filtered the data set based on the following keywords: ``Assoc Prof", ``Assoc Adj", ``Asst Adj", ``Asst Prof", ``Prof In", ``VIS Prof", ``Adj Instr", ``Lect", ``Grad", ``Adj". All other employees that didn't possess the listed keywords in the description of their position were identified as non-teaching staff. Further, log-transformed data was used to avoid the unnecessary computational burden.} For the purpose of this analysis, we considered testing the hypothesis defined in equation (\ref{hypo}) at a significance level of 5\% for the following three scenarios. In each scenario, we apply the proposed testing procedures for $t = 0.0, 0.2, 0.4, 0.6, 0.8, 1.0$, and obtain the corresponding test statistics and p-values.

\subsection{Using complete data to compare income distributions of faculty at CSU Monterey Bay and CSU San Bernardino} 

In this scenario, we examine the salaries of all instructional faculty from CSU Monterey Bay and CSU San Bernardino. Filtering data based on the employer name and the type of position resulted in obtaining a total of 546 records for CSU Monterey Bay faculty salaries and 1,265 records for CSU San Bernardino faculty. These two institutions were chosen because their Lorenz and generalized Lorenz curves appear to be very similar as shown in Figure \ref{app1}.    
\begin{figure}[H]
  \centering
  \includegraphics[width=1\textwidth]{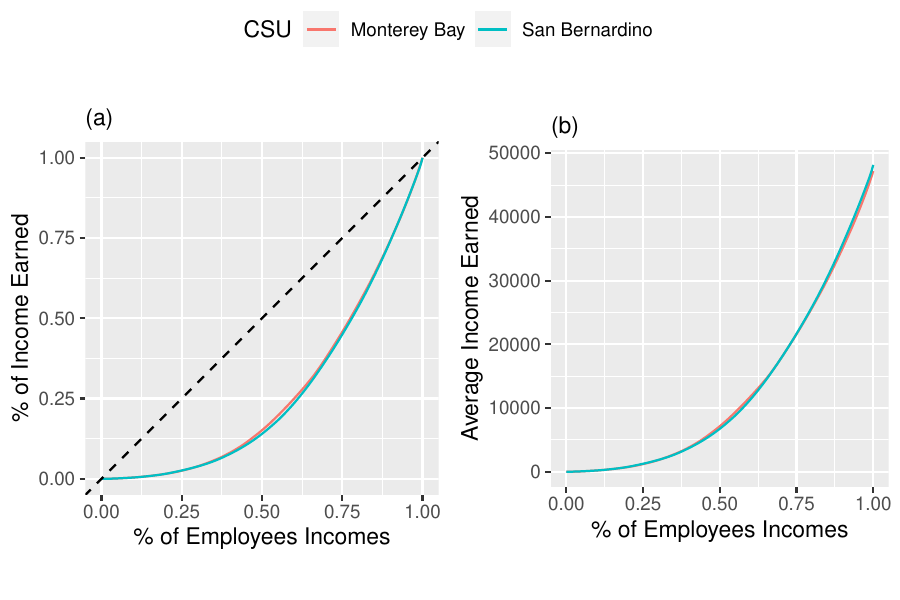}
  \caption{(a) Lorenz curves and (b) generalized Lorenz curves for salaries of CSU Monterey Bay and CSU San Bernardino faculty}
  \label{app1}
\end{figure}

{\setlength{\tabcolsep}{0.5em}
\begin{table}[H]
\footnotesize
\caption{Test statistic and p-value}
\centering
\begin{tabular}{ccrrrrrrr}
\hline
& &   \multicolumn{6}{c}{$t$}   \\  \hline
Method&  Value &  0.0 & 0.2  & 0.4  & 0.6  & 0.8  & 1.0 \\ \hline
ADF&	Test statistic& - &	0.0119	&0.0119&	0.0119&	0.0128&	0.0135 \\
&	p-value	& - & 0.9133&	 0.9133&	0.9133	&0.9100&	0.9076\\\hline
JEL&	Test statistic&	0.2116&	667.8884&	417.0143&	0.0380&	0.0070&	0.0138  \\
&	p-value	& 0.6455&	0.0000&	0.0000&	0.8454	&0.9332&	0.9065 \\ \hline
AJEL&	Test statistic&	0.2126&	669.2835&	417.9359&	0.0382&	0.0071&	0.0139 \\
&	p-value	&0.6448&	0.0000&	0.0000&	0.8450&	0.9331&0.9063 \\ \hline
\end{tabular}
\label{tab1}
\end{table}}
Table \ref{tab1} shows the computed test statistics and p-values for ADF, JEL, and AJEL methods at values of $t = 0.0, 0.2, 0.4, 0.6, 0.8, 1.0$. It is observed that both methods reject the null hypothesis at a 5\% significance level for $t=0.2$ and $t=0.4$, indicating sufficient evidence to conclude that the generalized Lorenz curves considered are significantly different for the 20th and 40th percentiles. However, for $t=0.0$, $t=0.6$, $t=0.8$, and $t=1.0$, the p-values obtained by both methods are greater than 0.05, implying that we do not have enough evidence to infer significant differences in the considered generalized Lorenz curves for the 0th, 60th, 80th, and 100th percentiles. These mixed results can be attributed to the multiple intersections of the considered generalized Lorentz curves.  

\subsection{Using complete data to compare income distributions of faculty at CSU San Bernardino and CSU San Francisco} 

In this scenario, we examine the salaries of all instructional faculty from CSU San Bernardino and CSU San Francisco. Filtering data based on the employer name and the type of position resulted in obtaining the total of 1,265 records for CSU San Bernardino faculty salaries and 2,294 records for CSU San Francisco faculty. These two institutions were chosen because their Lorenz and generalized Lorenz curves appear to be very distinct as shown in Figure \ref{app2}.

\begin{figure}[H]
  \centering
  \includegraphics[width=1\textwidth]{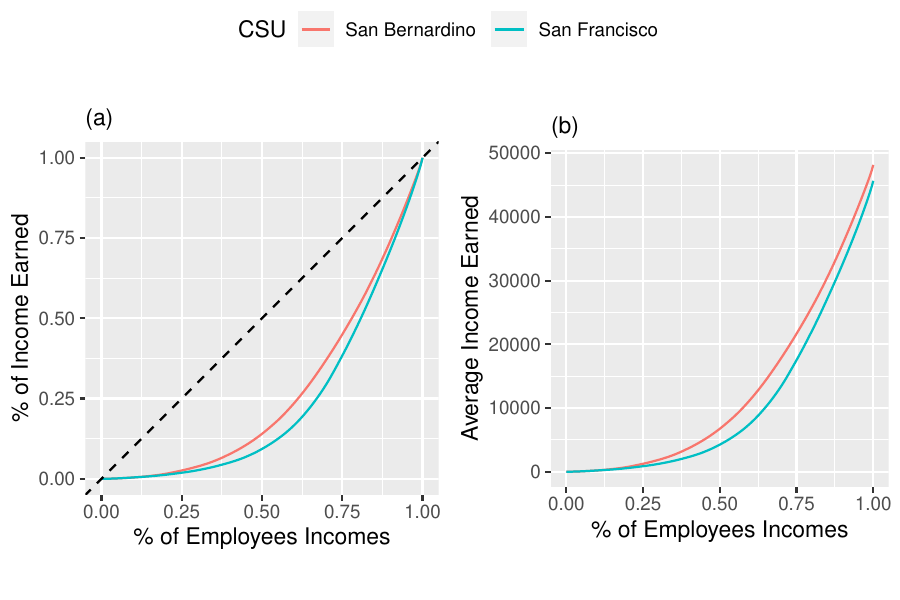}
  \caption{(a) Lorenz Curves and (b) Generalized Lorenz Curves for salaries of CSU San Bernardino and CSU San Francisco faculty}
  \label{app2}
\end{figure}

{\setlength{\tabcolsep}{0.5em}
\begin{table}[H]
\footnotesize
\caption{Test statistic and p-value}
\centering
\begin{tabular}{ccrrrrrrr}
\hline
& &   \multicolumn{6}{c}{$t$}   \\  \hline
Method&  Value &  0.0 & 0.2  & 0.4  & 0.6  & 0.8  & 1.0 \\ \hline
ADF	&Test statistic&	- &17.3485&	20.1831&	23.0600&	24.9908&	28.8444 \\
&	p-value&	- &0.0000&	0.0000&	0.0000&	0.0000&	0.0000 \\ \hline
JEL&	Test statistic&	1.5821&	217.8804&	942.6197&	250.0000&	30.2795&	20.0987  \\
&	p-value&	0.2085&	0.0000&	0.0000&	0.0000&	0.0000&	0.0000 \\ \hline
AJEL&	Test statistic&	 1.5860	&218.2452&	943.7490&	250.0000&	30.3477	&20.1438 \\
&	p-value&	0.2079 &	0.0000&	0.0000&	0.0000&	0.0000&	0.0000 \\ \hline
\end{tabular}
\label{tab2}
\end{table}}

Table \ref{tab2} presents the computed test statistics and p-values for ADF, JEL, and AJEL methods at $t = 0.0, 0.2, 0.4, 0.6, 0.8, 1.0$. The results indicate that both methods reject the null hypothesis for all values of $t$ except $t=0$, implying sufficient evidence to conclude that the two generalized Lorenz curves are significantly different for the 20th, 40th, 60th, 80th, and 100th percentiles at a 5\% significance level. However, for $t=0.0$, both methods yield p-values greater than 0.05, suggesting insufficient evidence to infer significant differences in the considered generalized Lorenz curves for the 0th percentile. These findings can be explained by the fact that the minimum salaries for both institutions are nearly identical, but the difference between the institutions increases as $t$ increases.

\subsection{Using incomplete data to compare income distributions of all faculty at CSU and UC}

In this scenario, we examine the instructional faculty salaries across all CSU and UC institutions. Filtering data based on the type of position resulted in obtaining the total of 34,927 records for CSU faculty salaries and 18,104 records for UC faculty salaries. The Lorenz and generalized Lorenz curves are shown in Figure \ref{app3}.
\begin{figure}[H]
  \centering
  \includegraphics[width=1\textwidth]{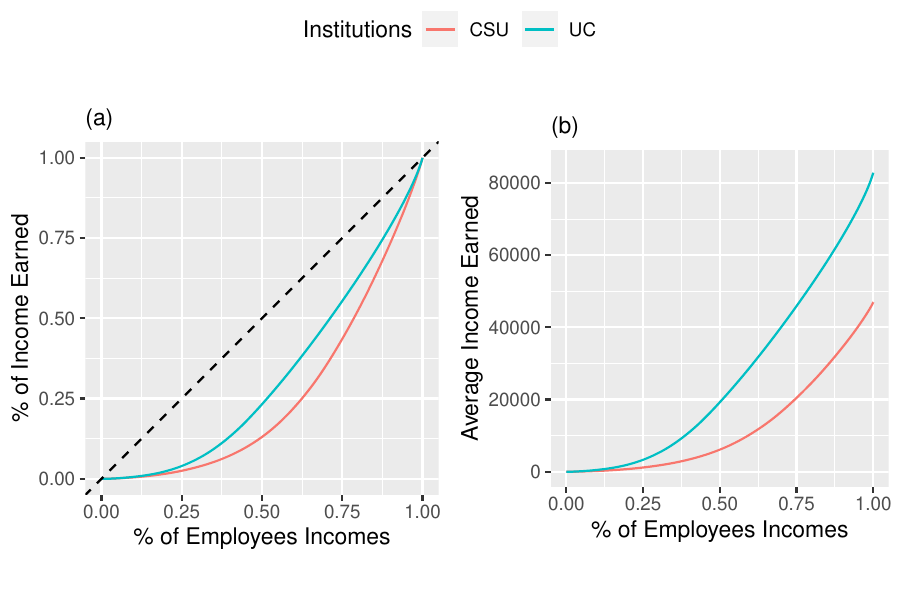}
  \caption{(a) Lorenz curves and (b) generalized Lorenz curves for salaries of CSU and UC faculty}
  \label{app3}
\end{figure}
Since both data sets are quite large, we decided to utilize the proposed procedures using the samples. As the empirical likelihood approach is known to be effective for relatively small samples, we selected sample sizes of $n_{1}=n_{2}=100$. Given that the difference between the generalized Lorenz curves is noticeable, we anticipate that both procedures will be capable of detecting the difference even with such modest sample sizes.

{\setlength{\tabcolsep}{0.5em}
\begin{table}[H]
\footnotesize
\caption{Test statistic and p-value}
\centering
\begin{tabular}{ccrrrrrrr}
\hline
& &   \multicolumn{6}{c}{$t$}   \\  \hline
Method&  Value &  0.0 & 0.2  & 0.4  & 0.6  & 0.8  & 1.0 \\ \hline
ADF&	Test statistic&	-&14.9673&	17.3378&	18.9842&	19.8941&	24.9073 \\
&	p-value&	-	& 0.0001& 0.0000&	0.0000&	0.0000&	0.0000 \\ \hline

JEL	&Test statistic&	0.6586&	48.6775	&20.5069&	66.4861	&38.4036&	25.2639 \\
&	p-value	&0.4171&	0.0000&	0.0000&	0.0000&	0.0000&	0.0000 \\ \hline
AJEL&	Test statistic&	0.6833&	49.7646&	20.9485&	67.5089	&39.0850&	25.7541 \\
&	p-value&	0.4084	&0.0000	&0.0000&	0.0000&	0.0000&	0.0000
\\ \hline
\end{tabular}
\label{tab3}
\end{table}}

Table \ref{tab3} displays the test statistics and p-values computed for ADF, JEL, and AJEL methods at $t=0.0, 0.2, 0.4, 0.6, 0.8,$ and $1.0$. The results indicate that all three methods reject the null hypothesis for all values of $t$ except $t=0$ at a 5\% significance level, providing sufficient evidence to conclude that the generalized Lorenz curves are significantly different for the 20th, 40th, 60th, 80th, and 100th percentiles. Conversely, for $t=0.0$, the p-values exceed 0.05, indicating insufficient evidence to conclude that the curves are significantly different at the 0th percentile. Similarly to the previous scenario, these findings can be attributed to the fact that the minimum salaries are nearly identical for both institutions, but the difference between them increases as $t$ increases. However, these findings are noteworthy because they were obtained using relatively small samples.

\subsection{Using complete data to compare income distributions of all faculty at CSUSB in Years 2009 and 2020}

In this scenario, we examine the salaries of all instructional faculty at CSUSB in the years 2009 and 2020. Filtering data based on the employer name and the type of position resulted in obtaining a total of 1,111 records for 2009 and 1,331 records for 2020. The Lorenz and generalized Lorenz curves are graphed in Figure \ref{app4}.

\begin{figure}[H]
  \centering
  \includegraphics[width=1\textwidth]{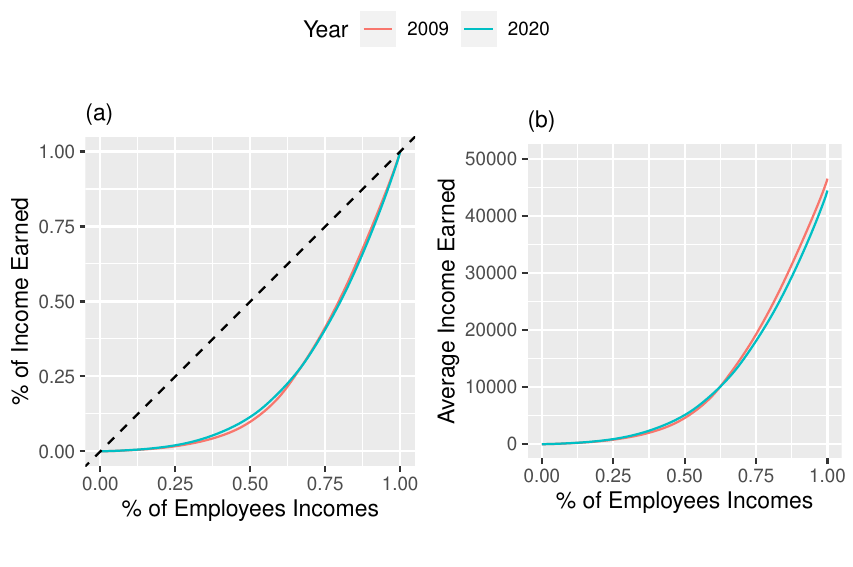}
  \caption{(a) Lorenz curves and (b) generalized Lorenz curves for salaries of CSU faculty in years 2009 and 2020}
  \label{app4}
\end{figure}

{\setlength{\tabcolsep}{0.5em}
\begin{table}[H]
\footnotesize
\caption{Test statistic and p-value}
\centering
\begin{tabular}{ccrrrrrrr}
\hline
& &   \multicolumn{6}{c}{$t$}   \\  \hline
Method&  Value &  0.0 & 0.2  & 0.4  & 0.6  & 0.8  & 1.0 \\ \hline
ADF&	Test statistic	& -&	1.5990&	2.2190	&2.2614&	2.5820&	2.5776 \\
&	p-value	&	-& 0.2060&	0.1363&	0.1326&	0.1081	&0.1084 \\ \hline
JEL	&Test statistic&	0.1719	&52.4042&29.2928	&0.1161&	0.0624&	0.1486 \\
&	p-value	&0.6785	&0.0000&	0.0000&	0.7333	&0.8028&	0.6999 \\ \hline
AJEL&	Test statistic&	0.1725&	52.5651&	29.3841&	0.1165&	0.0626&	0.1490 \\
&	p-value&	0.6779	&0.0000&	0.0000&	0.7329&	0.8025	&0.6994
\\ \hline
\end{tabular}
\label{tab4}
\end{table}}

Table \ref{tab4} shows the calculated test statistics and p-values for JEL and AJEL methods at $t = 0.0$, $0.2$, $0.4$, $0.6$, $0.8$, $1.0$ and for ADF method at $t = 0.2$, $0.4$, $0.6$, $0.8$, $1.0$. On the one hand, the ADF method led to the failure of rejection of the null hypothesis $\theta=0$ for all values of $t$, which means that at 5\% significance level, we do not have sufficient evidence to conclude that the considered generalized Lorenz curves are significantly different for 20th, 40th, 60th, 80th, and 100th percentiles. On the other hand, the JEL and AJEL methods result in large p-values for $t = 0.0$, $0.6$, $0.8$, $1.0$, and small p-values at $t=0.2$, $t=0.4$, which means that we have sufficient evidence to conclude that the generalized Lorenz curves in this scenario are significantly different at $t=0.2$ and $t=0.4$. Such mixed results can be explained by the fact that the considered generalized Lorentz curves intersect multiple times. The two proposed methods were able to capture these differences using relatively small samples.

\section{Conclusions}
In this paper, we developed two non-parametric JEL-based methods using a $U$-statistic to test the equality of two generalized Lorenz curves. The limiting distribution of the likelihood ratios is shown to follow a chi-squared distribution with one degree of freedom. Simulations studies with different distribution types and sample sizes illustrate that both methods show improved Type I error probability and power as the sample size increases. In general, the AJEL resulted in higher test powers in comparison to JEL across all distributions, sample sizes, and values of $t$, except for a few cases. However, the AJEL method has a higher Type I error rate than JEL, though still within an acceptable range. Moreover, these methods exhibit robustness across a range of scenarios, outperforming the existing ADF method. The proposed testing methods are applied to three distinct scenarios using  CSU and UC salary data. The results are found to be acceptable for both complete and incomplete data. This suggests that the proposed testing methods can be applied to a variety of data types, providing satisfactory results regardless of data completeness.

\newpage

\begin{appendices}

\section{Proofs of Theorems}

\begin{proof}\textbf{Theorem \ref{thm1}} \hfill \break

Let $n_{1}\leq n_{2}$. As shown in \cite{arv1969}, the jackknife procedure for the two sample $U$-statistics, $U_{n_{1},n_{2}}$, we have

\begin{equation}\label{A1}
\begin{aligned}
V_{i,0} &= n_{1}U_{n_{1},n_{2}}- (n_{1}-1) U_{n_{1}-1, n_{2}}^{-i,0}, \quad i = 1,\cdots,n_{1} \\
V_{0,j} &= n_{2}U_{n_{1},n_{2}}- (n_{2}-1) U_{n_{1}-1, n_{2}}^{0,-j}, \quad j = 1,\cdots,n_{2} 
\end{aligned}
\end{equation} 
Further, they proposed a consistent estimator of $Var(U_{n_{1},n_{2}})$ given as
\begin{equation*}\label{}
\begin{aligned}
\widehat{Var}_{\text{Jack}}(U_{n_{1},n_{2}}) = \dfrac{1}{n_{1}(n_{1}-1)} \sum_{i=1}^{n_{1}}\bigg( V_{i,0} - \Bar{V}_{\cdot, 0}\bigg)^{2} + \dfrac{1}{n_{2}(n_{2}-1)} \sum_{j=1}^{n_{2}}\bigg( V_{0,j} - \Bar{V}_{0, \cdot}\bigg)^{2},
\end{aligned}
\end{equation*} 
where $\Bar{V}_{\cdot, 0}$ and $\Bar{V}_{0, \cdot}$ are the means of $V_{i,0}$ and $ V_{0,j}$ respectively. Further, $U_{n_{1},n_{2}}$ is the original statistics based on all observations, $ U_{n_{1}-1, n_{2}}^{-i,0}$ is the statistics after leaving $X_{i}$ out, for $i =1,\dots,n_{1}$ and $U_{n_{1}-1, n_{2}}^{0,-j}$ is the statistics after leaving $Y_{j}$ out, for $j =1,\dots,n_{2}$.

\begin{Lemma} {(See \cite{arv1969})}
\begin{enumerate}
    \item Assume that $E|h(X,Y)|< \infty$, then $U_{n_{1},n_{2}} \xrightarrow[]{~~~\text{a.s.}~~~} \theta$ as $n \longrightarrow \infty$.
    \item Assume that $E[h^{2}(X,Y)]<\infty, \sigma_{1,0}^{2} >0$ and $\sigma_{0,1}^{2}>0$, let $S^{2}_{n_{1},n_{2}} = \dfrac{1}{n_{1}}\sigma_{1,0}^{2} + \dfrac{1}{n_{2}}\sigma_{0,1}^{2}$, then
\begin{equation*}\label{}
\begin{aligned}
\dfrac{U_{n_{1},n_{2}} - \theta}{S_{n_{1},n_{2}}}  &\xrightarrow[]{~~~\text{d}~~~} N(0,1) \quad \text{and} \quad
\widehat{Var}_{\text{Jack}}\big(U_{n_{1},n_{2}}\big) - S^{2}_{n_{1},n_{2}} =  o_{p}(n_{1}^{-1})  \quad \text{as} \quad n_{1}\longrightarrow \infty.
\end{aligned}
\end{equation*} 
\end{enumerate}
\end{Lemma}
In order to apply JEL, using (\ref{A1}), we can determine
\begin{equation*}\label{}
\begin{aligned}
V_{i,0} = \dfrac{1}{n_{2}}X_{i}\,I(X_{i} \leq \psi_{t}) - \dfrac{1}{n_{2}} \sum_{r = 1}^{n_{2}}Y_{r}\, I(Y_{r}\leq \psi_{t}), \quad i = 1,\dots,n_{1}
\end{aligned}
\end{equation*} 
and
\begin{equation*}\label{}
\begin{aligned}
V_{0,j} =  \dfrac{1}{n_{1}}\sum_{s = 1}^{n_{1}}X_{s}\,I(X_{s} \leq \psi_{t})  - \dfrac{1}{n_{1}}Y_{j}\,I(Y_{j} \leq \psi_{t}), \quad j = 1,\dots,n_{2}
\end{aligned}
\end{equation*} 
Let $n = n_{1}+n_{2}$. Consider
\begin{equation*}\label{}
\begin{aligned}
U_{n} = \dfrac{1}{n_{1}n_{2}} \sum_{1 \leq i \leq{n_1} < j \leq n}  \bigg(X_{i}\,I(X_{i} \leq \psi_{t})-Y_{j-n_{1}}\,I(Y_{j-n_{1}}\leq \psi_{t})\bigg)  
\end{aligned}
\end{equation*} 
and 
\begin{equation*}\label{}
\begin{aligned}
U_{n}^{-i} & = U\big(Z_{1}, Z_{2},\dots,Z_{i-1},Z_{i+1},\dots,Z_{n}\big) \\
&= {n-1 \choose 2}^{-1} \dfrac{1}{n_{1}n_{2}}  \sum_{\substack{1 \leq  r <  s  \leq n \\ r,s \neq i}}  \bigg(X_{r}\,I(X_{r} \leq \psi_{t}) -Y_{s-n_{1}}\,I(Y_{s-n_{1}}\leq \psi_{t})\bigg)  \\
&= \begin{cases} 
      \dfrac{n}{(n-2)} \bigg[ U_{n} - \dfrac{1}{n_{1}n_{2}} \mathlarger{\sum}_{i\leq{n_1}<j} \bigg(X_{i}\, I(X_{i} \leq \psi_{t})-Y_{j-n_{1}}\,I(Y_{j-n_{1}}\leq \psi_{t})\bigg)   \bigg] &,~~ 1 \leq i \leq n_{1} \\ \\
      \dfrac{n}{(n-2)} \bigg[ U_{n} - \dfrac{1}{n_{1}n_{2}} \mathlarger{\sum}_{j\leq{n_1}<i} \bigg(X_{j}\, I(X_{j} \leq \psi_{t})-Y_{i-n_{1}}\,I(Y_{i-n_{1}}\leq \psi_{t})\bigg)  \bigg] &,~~ n_{1} < i \leq n  
\end{cases}
\end{aligned}
\end{equation*} 
It can be seen that

\begin{equation*}\label{}
\begin{aligned}
\dfrac{1}{n_{1}n_{2}} \sum_{i\leq{n_1}<j} \bigg(X_{i}\,I(X_{i} \leq \psi_{t}) -Y_{j-n_{1}}\,I(Y_{j-n_{1}}\leq \psi_{t})\bigg)  = \dfrac{1}{n_{1}} V_{i,0}~, \quad 1 \leq i \leq n_{1}  \\
\end{aligned}
\end{equation*} 
and
\begin{equation*}\label{}
\begin{aligned}
\dfrac{1}{n_{1}n_{2}} \sum_{j\leq{n_1}<i} \bigg(X_{j}\,I(X_{j} \leq \psi_{t})-Y_{i-n_{1}}\, I(Y_{i-n_{1}}\leq \psi_{t})\bigg)  = \dfrac{1}{n_{2}} V_{0,i}~, \quad n_{1} < i \leq n  \\
\end{aligned}
\end{equation*}

Now, consider JEL given in (\ref{eq10}), for $1 \leq k \leq n$, we have
\begin{equation*}\label{}
\begin{aligned}
\widehat{V}_{k} &= nU_{n} - (n-1)U_{n-1}^{-k} \\
& = \dfrac{n(n-1)}{n-2}\Bigg[ \Bigg(\dfrac{V_{k,0}}{n_{1}}\Bigg)I(1 \leq k \leq n_{1}) + \Bigg(\dfrac{V_{0,k-n_{1}}}{n_{2}}\Bigg) I(n_{1} < k \leq n) \bigg] - \dfrac{n}{n-2} U_{n_{1},n_{2}} \\ 
\end{aligned}
\end{equation*}
Thus,
\begin{equation*}\label{}
\begin{aligned}
E \widehat{V}_{k} = \dfrac{n\theta}{n-2} \Bigg[ \Bigg( \dfrac{n_{2}-1}{n_{1}} \Bigg)I(1 \leq k \leq n_{1}) + \Bigg(\dfrac{n_{1}-1}{n_{2}}\Bigg)I(n_{1} < k \leq n) \Bigg]
\end{aligned}
\end{equation*} 
Under $H_{0}$, $E \widehat{V}_{k} = 0$. Next, following the similar arguments given in \cite{jing2009}, for fixed $t = t_{0}\in[0,1]$, it can be proven that $\ell(\theta(t_{0})) \longrightarrow \chi^{2}_{1}$, as $n_{1} \longrightarrow \infty$ . Thus, details are omitted here.

\end{proof}

\begin{proof}\textbf{Theorem \ref{thm2}} \hfill \break
The proof of this theorem is similar to Theorem 1 given in \cite{chen2008}. Let $\lambda^{\textnormal{Adj}}(t)$ be the solution to
\begin{equation}\label{eq15}
\begin{aligned}
\sum_{k=1}^{n+1}\frac{g^{\textnormal{Adj}}_{k}(t)}{1 + \lambda^{\textnormal{Adj}}(t)g^{\textnormal{Adj}}_{k}(t)} = 0.
\end{aligned}
\end{equation}
The first step is to show that $\lambda^{\textnormal{Adj}}(t) = O_{p}(n^{-1/2})$. By using Lemma 3 of \cite{owen1990} and the fact that $E(\hat{V}_{1}^{2}(t)) < \infty$, we can establish that $g^{*} = \max_{1\leq k \leq n} \norm{\hat{V}_{k}} = o_{p}(n^{1/2})$ and $\Bar{g}_{n}(t) = O_{p}(n^{-1/2})$. Let $\rho = \norm{\lambda^{\textnormal{Adj}}(t)}$, $a_{n} = o_{p}(n)$ and $\hat{\lambda}^{\textnormal{Adj}}(t) = \lambda^{\textnormal{Adj}}(t)/\rho$.  Multiplying $\hat{\lambda}^{\textnormal{Adj}}(t)/n$ to both sides gives

\begin{equation}\label{eq16}
\begin{aligned}
0 &= \frac{\hat{\lambda}^{\textnormal{Adj}}(t)}{n} \sum_{k=1}^{n+1}\frac{g^{\textnormal{Adj}}_{k}(t)}{1 + \lambda^{\textnormal{Adj}}(t) g^{\textnormal{Adj}}_{k}(t)} \\
&= \frac{\hat{\lambda}^{\textnormal{Adj}}(t)}{n} \sum_{k=1}^{n+1}g^{\textnormal{Adj}}_{k}(t) - \dfrac{\rho}{n}  \sum_{k=1}^{n+1} \frac{(\hat{\lambda}^{\textnormal{Adj}}(t) g^{\textnormal{Adj}}_{k}(t))^{2}}{1 + \rho \hat{\lambda}^{\textnormal{Adj}}(t) g^{\textnormal{Adj}}_{k}(t) } \\
& \leq \hat{\lambda}^{\textnormal{Adj}}(t)\Bar{g}_{n}(t)(1 - a_{n}/n) - \frac{\rho}{n(1+ \rho g^{*}(t))} \sum_{k=1}^{n} (\hat{\lambda}^{\textnormal{Adj}}(t) g^{\textnormal{Adj}}_{k}(t))^{2} \\
& = \hat{\lambda}^{\textnormal{Adj}}(t)\Bar{g}_{n}(t) - \frac{\rho}{n(1+ \rho g^{*}(t))} \sum_{k=1}^{n} \big(\hat{\lambda}^{\textnormal{Adj}}(t) g^{\textnormal{Adj}}_{k}(t)\big)^{2} + O_{p}(n^{-3/2}a_{n}).
\end{aligned}
\end{equation}
The inequality stated above is valid due to the non-negativity of the $(n + 1)$th term in the second summation. According to \cite{chen2008}, for any given $\epsilon > 0$, we have
\begin{equation}\label{eq19}
\begin{aligned}
\frac{1}{n} \sum_{k=1}^{n}\big(\lambda^{\textnormal{Adj}}(t) g^{\textnormal{Adj}}_{k}(t)\big)^{2} \geq  1- \epsilon.
\end{aligned}
\end{equation}

Therefore, as long as $a_{n} = o_{p}(n)$, equation (\ref{eq16}) implies that

\begin{equation}\label{eq20}
\begin{aligned}
\frac{\rho}{(1+\rho g^{*}(t))} \leq \hat{\lambda}^{\textnormal{Adj}}(t) \frac{\Bar{g}_{n}(t)(t)}{(1- \epsilon)}  =  O_{p}(n^{-1/2}).
\end{aligned}
\end{equation}
Thus, we get $\rho = O_{p}(n^{-1/2})$ and hence $\lambda^{\textnormal{Adj}}(t) = O_{p}(n^{-1/2})$. Now, consider
\begin{equation}\label{eq21}
\begin{aligned}
0 &= \frac{1}{n} \sum_{k=1}^{n+1}\frac{g^{\textnormal{Adj}}_{k}(t)}{1+ \lambda^{\textnormal{Adj}}(t) g^{\textnormal{Adj}}_{k}(t)} \\
&= \Bar{g}_{n}(t)(t) - \lambda^{\textnormal{Adj}}(t)\hat{V}_{n}(t) + o_{p}(n^{-1/2}),
\end{aligned}
\end{equation}
where $\hat{V}_{n} = (1/n) \sum_{k=1}^{n} g^{\textnormal{Adj}}_{k}(t)^{2}$. Hence, when $n \longrightarrow \infty, \lambda^{\textnormal{Adj}}(t) = \hat{V}_{n}^{-1} \Bar{g}_{n}(t) + o_{p}(n^{-1/2})$. Now, we expand $l^{*}(\theta(t))$ as follows

\begin{equation}\label{eq22}
\begin{aligned}
l^{*}(\theta(t)) &= \sum_{k=1}^{n+1} \log \big(1 + \lambda^{\textnormal{Adj}}(t) g^{\textnormal{Adj}}_{k}(t)\big) \\
&=   \sum_{k=1}^{n+1} \bigg\{\lambda^{\textnormal{Adj}}(t) g^{\textnormal{Adj}}_{k}(t) - \frac{\big(\lambda^{\textnormal{Adj}}(t)  g^{\textnormal{Adj}}_{k}(t)\big)^{2}}{2} \bigg \} + o_{p}(1).
\end{aligned}
\end{equation}
Substituting the expansion of $\lambda^{\textnormal{Adj}}$, we get that

\begin{equation}\label{eq23}
\begin{aligned}
-2l^{*}(\theta(t_{0})) &= n \hat{V}_{n}^{-1}\Bar{g}_{n}(t)^{2} + o_{p}(1) \\
& \xrightarrow[]{d}\chi^{2}_{1}.
\end{aligned}
\end{equation}
This completes the proof.
\end{proof}

\end{appendices}

\section*{References}
\renewcommand{\section}[2]{}
\bibliographystyle{apacite}
\setcitestyle{authoryear, open={((},close={))}}
\bibliography{References}

\end{document}